\documentclass[conference]{IEEEtran}
\usepackage{cite}

\ifCLASSINFOpdf
  \usepackage[pdftex]{graphicx}
  \usepackage{epstopdf}
  \DeclareGraphicsExtensions{.eps, .pdf, .jpeg} 
\else
\fi

\ifCLASSINFOpdf
	\usepackage[pdftex]{graphicx}
\else
	\usepackage[dvips]{graphicx}
\fi

\usepackage[cmex10]{amsmath}

\usepackage{array}
\usepackage{booktabs}
\usepackage{colortbl}
\usepackage{mdwmath}
\usepackage{amsfonts}
\usepackage{mdwtab}
\usepackage{eqparbox}
\usepackage[ruled,vlined]{algorithm2e}
\usepackage{flushend}   % equalizes columns in last page
\def\blfootnote{\xdef\@thefnmark{}\@footnotetext}    % To create unnumbered footnotes
\makeatother

\newtheorem{theorem}{Theorem}

\newtheorem{lemma}{Lemma}

\begin{document}

\title{Multi-User Coverage Probability of Uplink Cellular Systems: a Stochastic Geometry Approach}
\vspace{-6mm}
\author{\IEEEauthorblockN{F. Javier Martin-Vega, F. Javier Lopez-Martinez, Gerardo Gomez and Mari Carmen Aguayo-Torres}}
%\author{\IEEEauthorblockN{F. Javier Martin-Vega\IEEEauthorrefmark{1}, F. Javier Lopez-Martinez\IEEEauthorrefmark{2}, Gerardo Gomez\IEEEauthorrefmark{1}, Mari Carmen Aguayo-Torres\IEEEauthorrefmark{1}}
%% Delante mía, añade \IEEEauthorrefmark{1}\IEEEauthorrefmark{2}, detrás vuestra, solamente \IEEEauthorrefmark{1}
%\IEEEauthorblockA{\IEEEauthorrefmark{1}Dpto. Ingenier\'ia de Comunicaciones, Universidad de M\'alaga\\
%Campus de Teatinos, 29071 M\'alaga (Spain)\\Email\{fjmvega, ggomez, aguayo\}@ic.uma.es\\
%\IEEEauthorblockA{\IEEEauthorrefmark{2}Wireless Systems Lab, Electrical Engineering\\ Stanford University, Stanford, CA\\
%fjlm@stanford.edu }}
%}

\maketitle

\let\thefootnote\relax\footnote{
This work has been submitted to the IEEE for possible publication. Copyright may be transferred without notice, after which this version may no longer be accessible. \\
F. Javier Martin-Vega, Gerardo Gomez and Mari Carmen Aguayo-Torres are with the Departamento de Ingenier\'ia de Comunicaciones, Universidad de M\'alaga, 29071, M\'alaga, Spain (email: \{fjmvega, ggomez, aguayo\}@ic.uma.es). \\
F. Javier Lopez-Martinez is with the Wireless Systems Lab, Electrical Engineering, Stanford University, Stanford, CA (email: fjlm@stanford.edu). \\
This work has been partially supported by the Spanish Government and FEDER under projects TEC2010-18451 and COFUND2013-40259, the University of Malaga and the European Union under Marie-Curie COFUND U-mobility program (ref. 246550), and by the company AT4 wireless S.A.
}

\begin{abstract}
We analyze the coverage probability of multi-user uplink cellular networks with fractional power control. We use a stochastic geometry approach where the mobile users are distributed as a Poisson Point Process (PPP), whereas the serving base station (BS) is placed at the origin. Using conditional thinning, we are able to calculate the coverage probability of $k$ users which are allocated a set of orthogonal resources in the cell of interest, obtaining analytical expressions for this probability considering their respective distances to the serving BS. These expressions give useful insights on the interplay between the power control policy, the interference level and the degree of fairness among different users in the system.
\end{abstract}
\IEEEpeerreviewmaketitle

\section{Introduction}
\label{sec:Introduction}
% INTRODUCCION: SISTEMAS CELULARES Y FAINESS ENTRE UEs
Aiming to satisfy the ever-increasing demand for higher data rates, modern cellular technologies like Long Term Evolution (LTE) use aggressive frequency reuse policies, which have accentuated the problem of inter-cell interference compared to previous standards \cite{Himayat10}. This interference is highly dependent on the transmitted power of the different users, whose random positions and mobility affects the ability of the base stations (BS) to mitigate this problem. This causes huge differences on the received Signal to Interference plus Noise Ratio (SINR) due to path loss, being specially critical for cell-edge users, that tend to have a poorer performance compared to users located closer to the BS.

Each BS must also ensure a certain Quality of Service (QoS) for every user; hence, power control becomes a fundamental mechanism in the uplink (UL), as it impacts on the fairness among the users in the serving cell as well as on the level of interference caused to neighbor cells. Compared to the downlink (DL), the UL poses additional challenges since: (1) users positions are coupled with its serving BS, and (2) when power control is used, the interference level coming from a certain user depends not only on the distance of the BS to this user, but also on the distance between this interfering user and its serving BS. Additionally, even without power control, the interference behavior in the UL and DL is rather dissimilar. In the DL, those transmissions intended for cell-edge users tend to have stronger interference than for cell-interior ones, whereas in the UL all transmissions from the users inside the cell experience an interference with the same statistics.

%
%\subsection{Related Work}
% ESTADO DEL ARTE: Modelos tradicionales Wyner, hex. grid. Ventajas Stochastic Geometry. Andrews11. Novlan13 y uplink, dificultades, asunciones, etc. 

Stochastic geometry has emerged as a promising tool to analyze the performance of cellular systems, being an alternative to traditional approaches based on Wyner-type interference \cite{Wyner94} and hexagonal grid models \cite{Tukmanov13}, whose accuracy is known to be limited in different circumstances \cite{Xu11}. This approach typically considers the positions of transmitting nodes as a Poisson Point Process (PPP) where the receiver is placed at the origin \cite{Haenggi09} of a 2-D spatial grid. Despite being originally considered for ad-hoc and sensor networks due to the arbitrary positions of the nodes in such networks, the irregular cell patterns in modern cellular networks makes it the perfect technique to analyze their performance \cite{Andrews11}.

While most works based on random spatial models have focused on DL scenarios, their adequacy for modeling UL cellular networks has recently 
been addressed in \cite{Novlan13}. In this work, the authors provided the first known analytical results for the coverage probability of a typical user in a UL set-up, where fractional power control was implemented. As main assumptions, validated with realistic simulation models, they considered that the distances between interfering users and its serving BS are independent and identically distributed (i.i.d.), and that the BS falls in the Voronoi tessellation of each user. Based on this new approach, new analyses have been conducted in other UL scenarios involving fractional frequency reuse \cite{Novlan13b} or multi-tier cellular networks \cite{Elsawy14}.

Previous works in the literature are usually focused on only one active link between the transmitter and receiver nodes. Specifically, in \cite{Novlan13} their analysis considers the link between the serving BS of interest (placed at the origin) and a typical user. Since this randomly selected user can be located anywhere in the cell (cell interior, cell-edge, etc.), results are averaged over all spatial positions inside the cell. Although these results yield interesting insights on the performance of a typical user, they do not provide a clear understanding about the fairness among the users, or the performance of cell-edge users. Results concerning the coverage probability of UL cellular networks with multiple users are not available in the literature to the best of our knowledge.

In this paper, we present an analytical framework for the analysis of multi-user UL cellular systems with fractional power control, based on conditional thinning \cite{Dhillon13, Dhillon13b}. This technique has been used to model non-uniform user location distributions in DL transmissions \cite{Dhillon13} and different traffic load of each tier in heterogeneous networks \cite{Dhillon13b}. In our work, conditional thinning is used to obtain the set of interfering users for an arbitrary UL transmission allocated over one out of $k$ orthogonal resource groups. 

Using this new approach, the coverage probability of the $l^{th}$ user is obtained and ordered according to the distance from the user to the serving BS, which allocates $k$ orthogonal resource groups to users (${1\leq l\leq k}$). The joint distribution of the distances between the $l^{th}$ and ${k^{th}}$ users to the serving BS is also derived. Results give useful insights on the relation between power control and fairness among users. 

The rest of the paper is structured as follows. In Section \ref{sec:System Model}, we describe the system model and introduce our analytical framework based on conditional thinning. The main mathematical results are presented in Section \ref{sec:Math}, namely the joint distribution of the distances between the $l^{th}$ and ${k^{th}}$ users to the serving BS, and the multi-user coverage probability. Numerical results are given in Section \ref{sec:Numerical Results}, whereas main conclusions are drawn in section \ref{sec:Conclusions}.

\textit{Notation:} Throughout this paper, $|\cdot|$ stands for the Lebesgue measure, $\mathbb{E[\cdot]}$ for the expectation operator and $\mathbb{P[\cdot]}$ for a probability measure. Random Variables (RV) are represented with capital letters $X$ whereas deterministic variables are associated with lower case letters $x$. The conditional expectation of $X$ conditioned on $Y=y$ is denoted as $\mathbb{E}_{X|y} [X|y]$. $B(o, r)$ represents the closed ball centered at the origin $o$ being $r= \lVert x \rVert$ the distance from $x \in \mathbb{R}^2$ to $o$.

%A novel approach considers the user of {conditional thinning} in order to tackle more complex situations. For instance, in \cite{Dhillon13} conditional thinning is used to model non-uniform user location distributions. The proposed model for downlink analysis considers a uniform PPP that models BS locations whereas the user is placed at the origin. After selection of the serving BS, using distance as association metric, a conditional thinning is performed over the PPP except the serving BS location. After the thinning process the resulting PPP has the desired intensity and the user is shifted to the BS location. Additionally, in \cite{Dhillon13} conditional thinning is performed in order to model different traffic load of each tier in heterogeneous networks. The user connects to the strongest BS signal and then conditional thinning is performed  to model the probability of being active for each interfering BS.

\section{System Model}
\label{sec:System Model}
\subsection{System Model Description}
%
%A model for the analysis of UL cellular systems with fractional power control was introduced in \cite{Novlan13}. This model considered a uniform PPP of intensity $\lambda$ to represent user locations, whereas each BS was assumed to fall in the Voronoi cell of each user. For the sake of tractability the following three assumptions were made:
%
%\begin{itemize}%\compresslist
%  \item \emph{Assumption 1}: One BS falls in the Voronoi cell of each user. In that case, the distance $R$ to its serving BS is Rayleigh distributed with Probability Density Function (pdf): 
%\begin{equation}
%\label{PDF R}
%f_{R}(r) = 2 \pi \lambda r \mathrm{e}^{- \lambda \pi r^2}, \, r \geq 0
%\end{equation}
%
%\item \emph{Assumptions 2 and 3}: The distances between interfering users and its serving BS are i.i.d. The pdf of such distance is also Rayleigh distributed. 
%\footnote{Actually, in \cite{Novlan13} assumption 3 considers two models: regular and irregular grid. They assume a Rayleigh distribution for irregular BS deployment, which is the case considered in this paper.}
%\end{itemize}
%
In this paper we propose a system model that allows for a tractable analysis of multi-user UL scenarios with fractional power control, assuming one antenna at both transmitter and receiver sides. This model is illustrated in Fig. \ref{fig:System Model 1}.
%
  %This model is compared with a simulation model, which is not tractable, but provides realistic performance results. 
%It is considered the uplink of a cellular system employing one antenna at both transmitter and receiver with fractional power control. 
%
\begin{figure}[ht]
\centering
\includegraphics[width=2.75in]{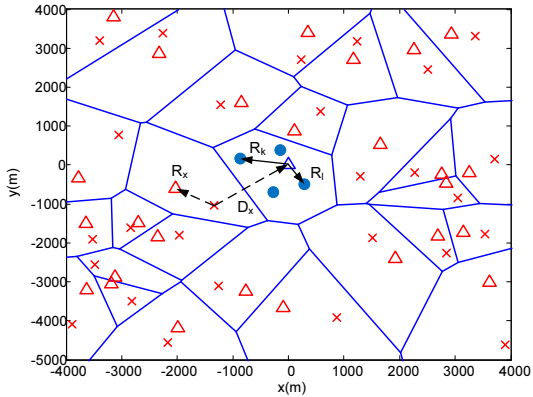}
\caption{System model of multi-user UL cellular system. BS are represented by triangles (blue: serving BS), users inside the serving cell are represented by blue dots, interfering users for the $l^{th}$ user transmission are depicted with a red cross. The distance from one interfering user to its serving BS ($R_x$) and to the target BS ($D_x$) are represented as an example.}
\label{fig:System Model 1}
\end{figure}

The target BS is considered to be placed at the origin, giving service to $k$ active users in $k$ orthogonal resource groups. Cells are assumed to be fully loaded, thus all available resource groups are used in the target and interfering cells and users are allocated a single resource group. 
These users, represented with blue dots, are ordered according to their distance to its serving BS, i.e. the origin. We focus on the $l^{th}$ user placed at distance $R_l$ from the origin with ${0 \leq R_1 \leq \cdots \leq R_l \leq \cdots \leq R_k}$. The BS positions of the interfering cells are indicated by red triangles, whereas the interfering users for the $l^{th}$ user data transmission are represented by red crosses. 

Since fractional power control is considered, the transmitted power depends on the distance between the user and its serving BS. This distance is represented as $R_x$ for an interfering user placed at $x \in \Phi_{i,l}$, where $\Phi_{i,l}$ denotes the random set of interfering user locations for $l^{th}$ user data transmission. Similarly, the distance between the interfering user located at $x$ and the target BS (i.e. the origin) is represented as $D_x$. 

Power loss due to propagation is modeled using a standard path loss model with $\alpha > 2$, whereas a Rayleigh model is assumed for small-scale fading. Fractional power control with parameter $\epsilon$ is assumed, hence the received signal power at distance $D_x$ from a user placed at distance $R_x$ from its BS is given by $G_x R_x^{\alpha \epsilon} D_x^{-\alpha}$, where $G_x$ is the fading coefficient that follows an exponential distribution with mean $1/\mu$. Thus, the SINR for the $l^{th}$ user data transmission follows the next expression
\begin{equation}
\label{eq:SINRl}
\text{SINR}_l = \frac{G_l R_l^{\alpha (\epsilon - 1)}}{I_l + \sigma^2}
\end{equation}
where $\sigma^2$ is the AWGN noise power and $I_l$ accounts for the interference experienced by the $l^{th}$ user transmission, given by 
\begin{equation}
\label{eq:Il}
I_l = \sum_{x \in \Phi_{i,l} } G_x R_x^{\alpha \epsilon} D_x^{-\alpha}
\end{equation}
It is important to note that in the UL, the interference suffered by all $k$ users transmission has the same statistics since interfering users positions scheduled at each resource group are expected to have the same distribution. Hence, from now on we will omit the sub-index $l$ in $\Phi_i$ for notation simplicity.
\subsection{Proposed Analytical Model}
The proposed model for multi-user uplink analysis is illustrated in Fig. \ref{fig:System Model 2}.
This model uses \textit{conditional thinning} in order to deal with multiple active links within the cell of interest.
\begin{figure}[ht]
\centering
\includegraphics[width=2.75in]{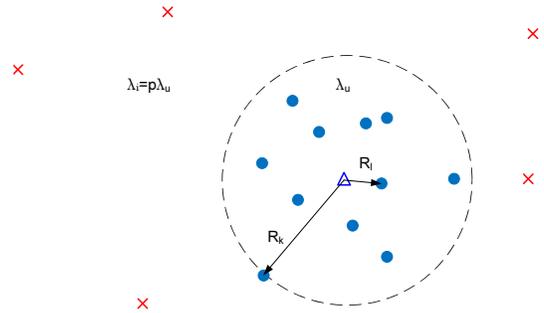}
\caption{Multi-User UL set-up based on conditional thinning for $k=11$. Interfering users for the $l^{th}$ user transmission are represented with red crosses.}
\label{fig:System Model 2}
\end{figure}
%
%This model uses \textit{conditional thinning} in order to deal with multiple active links within the cell of interest. This technique has been used to model non-uniform user location distributions in DL transmissions \cite{Dhillon13}, considering a uniform PPP that models BS locations with the desired user placed at the origin: after selecting the serving BS using distance as association metric, conditional thinning is performed over the PPP except on the serving BS location. After the thinning process, the resulting PPP has the desired intensity and the user is shifted to the BS location. Conditional thinning was also used in \cite{Dhillon13b} in order to model different traffic load of each tier in heterogeneous networks: the user connects to the strongest BS signal and then conditional thinning is performed to model the probability of being active for each interfering BS.

Let us consider the target BS to be placed at the origin and an uniform PPP $\Phi$ of intensity $\lambda$ over $\mathbb{R}^2$ that represents the set of active users. 
We use conditional thinning as follows:

First, the $k$ nearest points of $\Phi$ to the origin are selected. These points represent the locations of $k$ users scheduled in $k$ orthogonal resource groups. Then, thinning with probability $p$ is performed to all points except those $k$ inside the closed ball $B(o,r_k)$, being $r_k$ the distance to the $k^{th}$ point.

The resulting set of points outside the ball $B(o,r_k)$ is a non-uniform PPP $\Phi_i$ of intensity measure $\Lambda_i(A) = p  \lambda |A \setminus B(o,r_k)|$ \cite{Stoyan13}. Such random set of points represents the interfering user locations for the $l^{th}$ user data transmission. Since these interfering users are using one of $k$ available resources, we choose the thinning probability to be $p=1/k$. As the model considers that there is only one user scheduled per orthogonal resource group per cell, the intensity of BS is exactly the same as the intensity of interfering users. The random set of $k$ points around the origin $\Phi_d$ has an intensity measure $\Lambda_d (A) = \lambda |A \cap B(o,r_k)|$. Hence, the complete set of user locations $\Phi_u$ is given by the set of user locations within the target cell (using all available resource groups) and the set of interfering users scheduled in one resource group, i.e. $\Phi_u = \Phi_d + \Phi_i$.

As in \cite{Novlan13} distances $\{R_x\}$ from each interfering user to its serving BS are assumed to be i.i.d. RV following Rayleigh distributions with 
\begin{equation}
\label{eq:PDF Rx}
   f_{R_x} (r_x) = 2 \pi p \lambda r_x \mathrm{e}^{-p \lambda  \pi r_x^2}, \, r_x \geq 0
\end{equation}
Hence, notice that the proposed model is equivalent to the model presented in \cite{Novlan13} for $p=1$, and $k=l=1$.
\subsection{Simulation Model}
In order to asses the validity of the proposed analysis model, we also introduce a more realistic model for simulation. A uniform PPP of intensity $\lambda_b$ representing the BS locations is first considered. Since in the analysis model the intensity of BS is the same as that of interfering users,   we use $\lambda_b = \lambda/k$ aiming to compare the results of both models.

%In order to asses the validity of the proposed analytical model, we also introduce a realistic model for simulation. This model considers a uniform PPP of intensity $\lambda_b$ that represents the BS locations. Since the intensity of BS (interfering users) in the analytical model is $\lambda p$ and $p=1/k$ we use $\lambda_b = \lambda/k$ aiming to compare the results of both models.

The association between user and BS is based on distance, hence the Voronoi tessellation is performed where one randomly chosen point is the target BS. Then, $k$ points representing the $k$ active users are placed randomly inside the target cell, whereas only one user is placed in each interfering cell. Notice that both sets of points, active users inside the target cell and interfering users, are not a PPP. To explain that, recall that the number of points falling in a Voronoi cell tends to be higher as the cell is bigger; in our case, one interfering user falls in any cell independently of its size. 
\section{Mathematical Results}
After presenting the analytical framework for the analysis of multi-user UL cellular networks, we now present the main mathematical contributions of this paper. First, we derive the joint distribution of the distances between the $l^{th}$ and $k^{th}$ users and the serving BS. Then, we use this result to calculate the coverage probability of the $l^{th}$ user in the investigated scenario.
\vspace{-6mm}
\label{sec:Math}
\subsection{Joint Distribution of Distances}
\label{sec:Joint Distribution of Distances}
In the analytical model, the $k$ users of interest are ordered according to their distances to the serving BS (i.e., the origin), and the interfering users are located at a distance greater than $R_k$. This interdependence affects the distribution of the SINR for the $l^{th}$ user transmission, due to the inherent correlation between 
$R_l$ and $R_k$. In the next lemma, we calculate their joint pdf.

\begin{lemma}
The joint pdf of $R_l$ and $R_k$ with $0<l<k$ is
\begin{equation}
\label{eq:Definition of PDF}
f_{R_l,R_k}(r_l, r_k) = \frac{4 \mathrm{e}^{-\pi r_k^2 \lambda} (\lambda \pi)^k r_k
 r_l^{2l-1} (r_k^2 - r_l^2)^{k-l-1} }{(k-l-1)! (l-1)!}
\end{equation}
where $0 \leq r_l \leq r_k$.
\end{lemma}
\begin{IEEEproof}
The calculation of the joint pdf follows a similar procedure as in \cite{Dhillon13}. Hence, we define disjoint sets in order to use the independence property of the PPP. Let us consider the next disjoint sets
\begin{align}
%\label{eq:Disjoint Sets}
& \Psi_1 = \lbrace x \in \mathbb{R}^2: \lVert  x \rVert \leq r_l \rbrace \nonumber \\
& \Psi_2 = \lbrace x \in \mathbb{R}^2: r_l < \lVert  x \rVert \leq r_l+\mathrm{d}r_l \rbrace \nonumber \\
& \Psi_3 = \lbrace x \in \mathbb{R}^2: 
   r_l + \mathrm{d}r_l < 
   \lVert  x \rVert \leq r_k \rbrace \nonumber \\
& \Psi_4 = \lbrace x \in \mathbb{R}^2: 
   r_k < \lVert  x \rVert \leq r_k + \mathrm{d}r_k \rbrace 
\end{align}
The joint pdf of $R_l$ and $R_k$ with $0<l<k$ is by definition 
\begin{equation}
\label{eq:Definition of PDF}
f_{R_l,R_k}(r_l, r_k) = \lim_{\overset{\mathrm{d}r_l \to 0} {\underset{\mathrm{d}r_k \to 0}{}} } \frac{\mathbb{P} \lbrace R_l \in \Psi_2, 
   R_k \in \Psi_4 \rbrace }{\mathrm{d}r_l \mathrm{d}r_k}
\end{equation}
Notice that the numerator can be expressed as follows:
\begin{align}
\label{eq: Numerator Joint PDF}
\mathbb{P}& \lbrace R_l \in \Psi_2, R_k \in \Psi_4 \rbrace = \nonumber \\
   & \mathbb{P} \lbrace \Phi (\Psi_1) = l-1 \rbrace 
   \cdot \mathbb{P} \lbrace \Phi (\Psi_2) = 1 \rbrace \cdot \nonumber \\
     & \mathbb{P} \lbrace \Phi (\Psi_3) = k-l-1 \rbrace 
   \cdot \mathbb{P} \lbrace \Phi (\Psi_4) = 1 \rbrace 
\end{align}
being $\Phi(\Psi)$ a random counting measure of a Borel set $\Psi$. Since $\Phi$ is a uniform PPP, $\Phi(\Psi)$ follows Poisson distribution with mean $\lambda \lvert \Psi \rvert$ \cite{Stoyan13}. Substituting the probability of each event in (\ref{eq: Numerator Joint PDF}) and calculating the limits in (\ref{eq:Definition of PDF}) yields the desired pdf.
\end{IEEEproof}
Figs. \ref{fig:Joint PDF l=2 k=4} and \ref{fig:Joint PDF l=2 k=50} illustrate the joint pdf of the distances for the second and the $k^{th}$ user, when $k=4$ and $k=50$, respectively. The correlation is more noticeable when $l$ and $k$ have similar values. 
\begin{figure}[ht]
\centering
\includegraphics[width=2.75in]{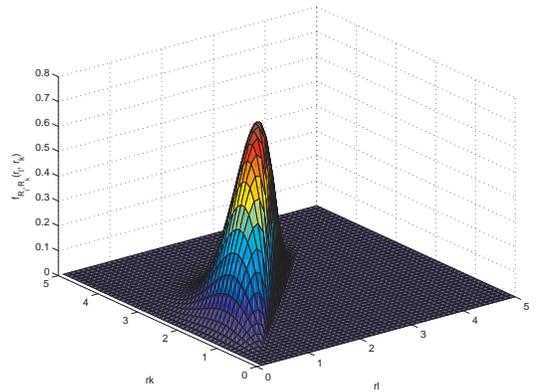}
\caption{ Joint pdf for $l=2$ and $k=4$ with $\lambda=0.24$ }
\label{fig:Joint PDF l=2 k=4}
\end{figure}
\begin{figure}[ht]
\centering
\includegraphics[width=2.75in]{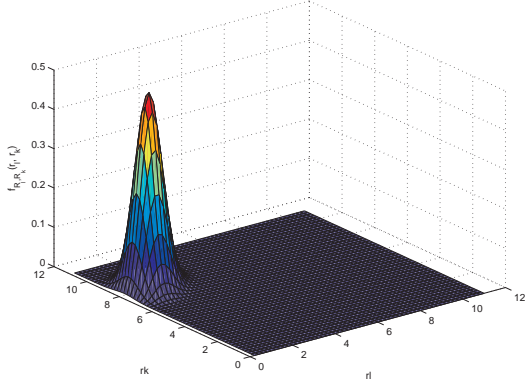}
\caption{ Joint pdf for $l=2$ and $k=50$ with $\lambda=0.24$ }
\label{fig:Joint PDF l=2 k=50}
\end{figure}
\vspace{-2mm}
\subsection{Multi-User Coverage Probability}
\label{sec:Coverage Probability}
The coverage probability represents the probability for a user to have a SINR higher than certain threshold $t$. The main result is stated in Theorem \ref{theor:Multi-user coverage}, which corresponds to the coverage probability of the the $l^{th}$ user.
\begin{theorem}[Multi-user coverage probability]
\label{theor:Multi-user coverage}
The coverage probability of the $l^{th}$ user considering a system with $k$ orthogonal resource groups that are distributed among $k$ active users with $l < k$ is given by: 
\begin{align}
p_c(l,k,t,\lambda,& p,\alpha, \epsilon, \mu, \sigma^2) \nonumber \\
   &= \mathbb{E}_{R_l,R_k} \left[ \xi (r_l, r_k) \right] \nonumber \\
   &= \int_{0}^{\infty} \int_{r_l}^{\infty} \xi (r_l, r_k) f_{R_l,R_k}(r_l, r_k) \mathrm{d}r_k \mathrm{d}r_l
\end{align}
where $f_{R_l,R_k}(r_l, r_k)$ is the joint pdf of distances and
\begin{equation}
   \xi (r_l, r_k) = \mathrm{e}^{-\mu t \sigma^2 r_l^{\alpha(1-\epsilon)}} 
   \mathcal{L}_{I_l|r_l,r_k}(\mu t r_l^{\alpha(1-\epsilon)})
\end{equation}
being $\mathcal{L}_{I_l|r_l,r_k}(s)$ the Laplace transform of the interference conditioned on $r_l$ and $r_k$. This term evaluated at 
$s=\mu t r_l^{\alpha(1-\epsilon)}$ has the following expression
\begin{align}
   & \mathcal{L}_{I_l|r_l,r_k}(\mu t r_l^{\alpha(1-\epsilon)}) = \nonumber \\
   & \mathrm{exp}\left( -2 \pi p \lambda \int_{r_k}^{\infty} 
      \left(1 - \int_{0}^{\infty} 
      \frac{\pi p \lambda \mathrm{e}^{-p \lambda \pi q }}
      {1 + t r_l^{\alpha(1-\epsilon)} q^{\alpha \epsilon / 2} v^{-\alpha} } \mathrm{d}q
      \right)v \mathrm{d}v\right)
\end{align}
\end{theorem}
\begin{IEEEproof}
The coverage probability for the $l^{th}$ user can be expressed as
\begin{align}
\label{eq:Multi - User Coverage Probability}
p_c(& l,k, t, \lambda,p,\alpha, \epsilon, \mu, \sigma^2) 
   = \mathbb{P}\left[ \mathrm{SINR}_l > t\right] \nonumber \\
   &\overset{(a)}{=} \int_0^\infty \mathbb{P}
      \left[\mathrm{SINR}_l > t \rvert r_l\right] 
         f_{R_l}(r_l) \mathrm{d}r_l \nonumber \\
   &= \int_0^\infty 
      \mathbb{P}\left[G_l > t(I_l + \sigma^2)r_l^{\alpha (\epsilon - 1)} \rvert
       r_l\right] 
      f_{R_l}(r_l) \mathrm{d}r_l \nonumber \\
   &\overset{(b)}{=} \int_0^\infty 
      \mathbb{E}_{I_l} 
      \left[
      \mathbb{P}\left[G_l > t(i_l + \sigma^2)r_l^{\alpha (\epsilon - 1)} 
         \rvert r_l,i_l\right]
      \right]
      f_{R_l}(r_l) \mathrm{d}r_l \nonumber \\
   &\overset{(c)}{=}  \int_0^\infty 
      \mathrm{e}^{-\mu t \sigma^2 r_l^{\alpha (\epsilon - 1)}}
      \mathbb{E}_{I_l|r_l} 
      \left[
      \mathrm{e}^{-\mu t I_l r_l^{\alpha (\epsilon - 1)}} 
         \rvert r_l\right]
      f_{R_l}(r_l) \mathrm{d}r_l      
\end{align}
%    = \int_0^\infty \mathbb{P}\left[
%      \frac{G_l R_l^{\alpha (\epsilon - 1)}}{I_l + \sigma^2} 
%      > t | r_l \right] f_{R_l}(r_l) \mathrm{d}r_l \nonumber \\
where $(a)$ and $(b)$ follow from the total probability theorem \cite{Papoulis}, while $(c)$ follows from the fact that $G_l$ has an exponential distribution with mean $1/\mu$. 

The term $\mathcal{L}_{I_l|r_l}(s) = \mathbb{E}_{I_l|r_l} 
      \left[
      \mathrm{e}^{I_l} 
         \rvert r_l\right]$
represents the Laplace transform of the interference conditioned on $r_l$. The RV $R_l$ and $R_k$ are correlated as $R_l \leq R_k$. Since $I_l$ depends on $R_k$ due to the fact that the interfering users are placed farther than $R_k$, the RV $I_l$ also depends on $R_l$. Hence we have to deal with such dependence as follows
\begin{align}
\label{eq:Laplace cond rl}
\mathcal{L}_{I_l|r_l}(s) & 
   = \mathbb{E}_{I_l|r_l, r_k} 
      \left[ 
         \mathbb{E}_{R_k}\left[  \mathrm{e}^{-s I_l} \rvert r_l,r_k \right]
      \right] \nonumber \\
   &= \mathbb{E}_{I_l|r_l, r_k} 
      \left[ 
           \int_{r_l}^{\infty} \mathrm{e}^{-s I_l} f_{R_k \lvert r_l}(r_k) 
           \mathrm{d}r_k  \rvert r_l,r_k
      \right] \nonumber \\
   &= \int_{r_l}^{\infty} 
      \mathbb{E}_{I_l|r_l, r_k} \left[ \mathrm{e}^{-s I_l} \rvert r_l,r_k \right] 
      f_{R_k \lvert r_l}(r_k) \mathrm{d}r_k  
\end{align}
where the total probability theorem and linearity of expectation operator have been used. 

The term  $\mathcal{L}_{I_l|r_l, r_k}(s)  = \mathbb{E}_{I_l|r_l, r_k} \left[ \mathrm{e}^{-s I_l} \rvert r_l,r_k \right]$ stands for the Laplace transform of the interference conditioned on $r_l$ and $r_k$ and can be expressed as
\begin{align}
\label{eq:Laplace cond rl and r_k}
&\mathcal{L}_{I_l|r_l, r_k}(s) = \mathbb{E}_{\Phi_i, \{G_x\} } 
   \left[
   \mathrm{e}^{-s \sum_{x \in \Phi_i} G_x R_x^{\alpha \epsilon} D_x^{-\alpha}} 
   \right] \nonumber \\
%   &\overset{(a)}{=} \mathbb{E}_{\Phi_i  } 
%   \left[
%    \prod_{x \in \Phi_i} \mathbb{E}_{G_x} \left[
%       \mathrm{e}^{-s G_x R_x^{\alpha \epsilon} D_x^{-\alpha}} \right]
%   \right] \nonumber \\
   &\overset{(a)}{=} \mathbb{E}_{R_x, D_x} 
   \left[
    \prod_{x \in \Phi_i} \mathbb{E}_{G_x} \left[
       \mathrm{e}^{-s G_x R_x^{\alpha \epsilon} D_x^{-\alpha}} \right]
   \right] \nonumber \\
%   &\overset{(c)}{=} \mathbb{E}_{D_x} 
%   \left[
%    \prod_{x \in \Phi_i} \mathbb{E}_{R_x} \left[ \mathbb{E}_{G_x} \left[
%       \mathrm{e}^{-s G_x R_x^{\alpha \epsilon} D_x^{-\alpha}} \right] \right]
%   \right] \nonumber \\
   &\overset{(b)}{=} \mathbb{E}_{D_x} 
   \left[
    \prod_{x \in \Phi_i} \mathbb{E}_{R_x} \left[ 
    \frac{\mu}{\mu + s R_x^{\alpha \epsilon} D_x^{-\alpha}}
   \right] 
   \right] \nonumber \\
   &\overset{(c)}{=} \exp \left( -2 \pi \lambda p 
   \int_{r_k}^{\infty} \left(1 
   -  \int_{0}^{\infty} 
      \frac{\pi p \lambda \mathrm{e}^{-p \lambda \pi q } \mathrm{d}q}
      {1 + t r_l^{\alpha(1-\epsilon)} q^{ \frac{\alpha \epsilon }{2}} v^{-\alpha} }      
   \right) v \mathrm{d}v
   \right) 
\end{align} 
where the dependence with $R_l$ and $R_k$ resides in the non-uniform PPP $\Phi_i$ since its intensity is $\Lambda_i(A) = p  \lambda |A \setminus B(o,r_k)|$. Step $(a)$ comes from the fact that the fading is independent of the PPP, $(b)$ comes from the independence assumption between $R_x$ and $D_x$ and $(c)$ from the Probability Generating Functional (PGFL) \cite{Stoyan13} and the assumption of $R_x$ following a Rayleigh distribution as in \cite{Novlan13}.

Substituting (\ref{eq:Laplace cond rl and r_k}) and (\ref{eq:Laplace cond rl}) with $s=\mu t r_l^{\alpha(1-\epsilon)}$ in (\ref{eq:Multi - User Coverage Probability}) and taking into account that the conditional pdf $f_{R_k \lvert r_l}(r_k)$ can be obtained from the joint pdf and the marginal pdf of $R_l$ as $f_{R_k \lvert r_l}(r_k) = f_{R_l,R_k}(r_l, r_k)/f_{R_l}(r_l)$, the proof is completed.
\end{IEEEproof}
Theorem \ref{theor:Multi-user coverage} provides the coverage probability of the $l^{th}$ user with $l<k$. The following lemma gives the coverage probability for the cell-edge user. 
\begin{lemma}
The coverage probability of the $k^{th}$ user follows the next expression
\begin{equation}
p_c(k,t,\lambda, p,\alpha, \epsilon, \mu, \sigma^2) =
   \int_{0}^{\infty} \xi (r_k) f_{R_k}(r_k) \mathrm{d}r_k
\end{equation}
where $f_{R_k}(r_k)$ is the marginal pdf distribution of the $k^{th}$ nearest point \cite{Haenggi05} given by
\begin{equation}
f_{R_k}(r_k) = 2 \frac{(\lambda \pi)^{k}}{(k-1)!} r_k^{2k-1} \mathrm{e}^{-\lambda \pi r_k^2}
\end{equation}
and 
\begin{equation}
\xi (r_k) = \mathrm{e}^{-\mu t \sigma^2 r_k^{\alpha(1-\epsilon)}} 
   \mathcal{L}_{I_k|r_k}(\mu t r_k^{\alpha(1-\epsilon)})
\end{equation}
where $\mathcal{L}_{I_k|r_k}(\mu t r_k^{\alpha(1-\epsilon)})$ the Laplace transform of the interference affecting the $k^{th}$ user transmission conditioned on $r_k$, given by
\begin{align}
   & \mathcal{L}_{I_k|r_k}(\mu t r_k^{\alpha(1-\epsilon)}) = \nonumber \\
   & \mathrm{exp}\left( -2 \pi p \lambda \int_{r_k}^{\infty} 
      \left(1 - \int_{0}^{\infty} 
      \frac{\pi p \lambda \mathrm{e}^{-p \lambda \pi q }}
      {1 + t r_k^{\alpha(1-\epsilon)} q^{\alpha \epsilon / 2} v^{-\alpha} } \mathrm{d}q
      \right)v \mathrm{d}v\right)
\end{align}
\end{lemma}
\begin{proof}
The proof is analogous to Theorem \ref{theor:Multi-user coverage} except from the fact that the SINR of the $k^{th}$ user transmission only depends on the distance to the origin of one particular user; note that when $l<k$ the SINR depends both on $R_l$ and $R_k$. Hence, the Laplace transform of the interference only depends on $R_k$ and only the marginal pdf of $R_k$ is necessary.
\end{proof}

\section{Numerical Results}
\label{sec:Numerical Results}
\subsection{Coverage probability}
We now evaluate the expressions for the coverage probability previously derived, and compare these results with our simulation model. Different values of the power control factor $\epsilon$ are used so as to provide a clear understanding of the relation between power control and fairness among users.

Fig. \ref{fig:fig_V_1} shows the coverage probability considering different numbers of orthogonal resources per cell, i.e. $k = \{10, 25, 50\}$, assuming a full power control policy ($\epsilon=1$). We see how the coverage probability is the same for all $k$ scheduled users, i.e. it does not depend on $l$ for both analytical and simulation models. This is coherent with the fact that full compensation of path loss makes all user transmissions to have the same average received power. Since the interference experienced by all user transmissions is the same, the coverage is also the same. Hence, in this case the fairness between users is maximal. We also observe how the analytical model provides slightly more pessimistic results than the simulation model.
\begin{figure}[ht]
\centering
\includegraphics[width=2.75in]{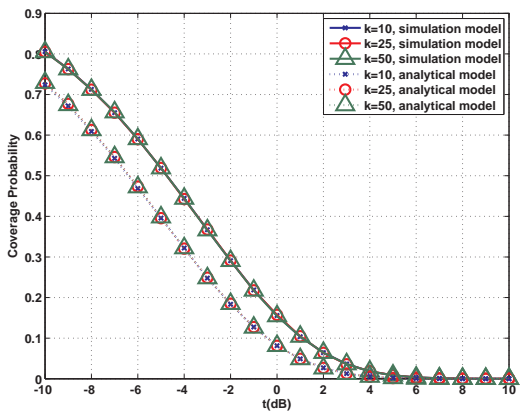}
\caption{ Coverage probability for $k=\{10, 25, 50\}$ with full power control ($\epsilon=1$), without noise, $\alpha = 2.5$, $\lambda_b=0.24$}
\label{fig:fig_V_1}
\end{figure}

Fig. \ref{fig:fig_V_2} illustrates the coverage probability for cell-interior ($l=1$) and cell-edge ($l=k$) users with $k=25$ and a power control factor $\epsilon=0.75$. We observe how both analytical and simulation models still behave quite close to each other. In both models, since the compensation of path loss is not total, transmissions from users closer to the BS are associated to higher SINR values than those in the cell-edge, so there exists a difference in coverage between users.
\begin{figure}[ht]
\centering
\includegraphics[width=2.75in]{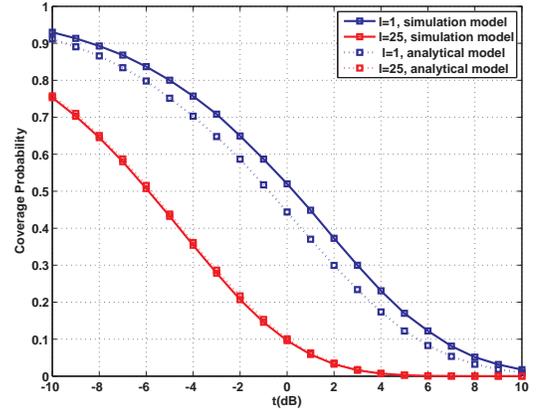}
\caption{ Coverage probability for cell-interior $l=1$ and cell-edge $l=k$ with $k=25$, $\epsilon=0.75$, without noise, $\alpha = 2.5$, $\lambda_b=0.24$}
\label{fig:fig_V_2}
\end{figure}

Fig. \ref{fig:fig_V_3} shows the coverage probability in the absence of power control, which corresponds to the worse case in terms of fairness. Hence, we observe  that the difference in coverage between cell-interior $l=1$ and cell-edge $l=k$ users is maximal. We also see how for the cell-edge user the analytical model yields a coverage significantly greater than the simulation model. The reason behind that is related to the different distribution of points used to model active user locations in both models. As mentioned in section \ref{sec:System Model}, in the analytical model user locations form a PPP, whereas in the simulation model this does not hold. This issue has a significant impact on the pdf of the distances $R_l$ specially for cell-edge users, and is addressed in detail in the next subsection. 
\begin{figure}[ht]
\centering
\includegraphics[width=2.75in]{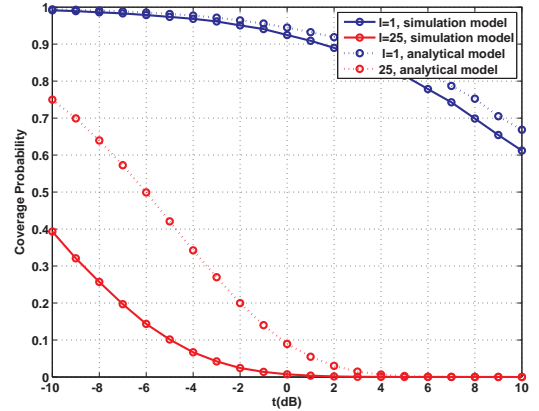}
\caption{ Coverage probability for cell-interior $l=1$ and cell-edge $l=k$ with $k=25$, without power control $\epsilon=0$, without noise, $\alpha = 2.5$, $\lambda_b=0.24$ }
\label{fig:fig_V_3}
\end{figure}
% In the next subsection these details are going to be explained.
%
%\vspace{-2mm}
\subsection{Marginal distributions of distances}
One of the assumptions of the proposed model follows from \cite{Novlan13} and states that $R_x$ with $x \in \Phi_i$ are i.i.d. Rayleigh distributed RVs. Fig. \ref{fig:fig_V_4} shows the theoretical (Rayleigh) distribution used in the analytical model and the empirical distribution obtained from the simulation model. We observed that both pdfs are quite similar, so it is expected that the statistics of the transmitted power of the interfering users are also close to each other.
%\vspace{-2mm}
\begin{figure}[ht]
\centering
\includegraphics[width=2.75in]{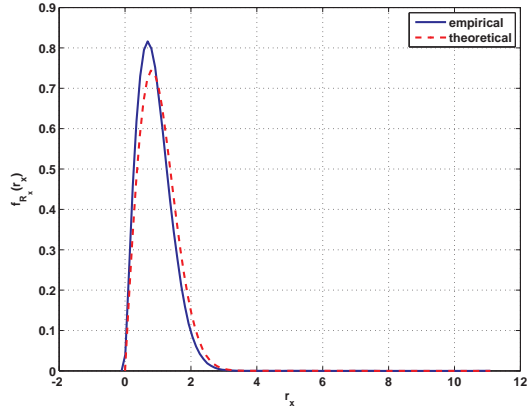}
\caption{Empirical and theoretical pdf of $R_x$.}
\label{fig:fig_V_4}
\end{figure}
\vspace{-2mm}

Fig. \ref{fig:fig_V_5} shows the marginal pdfs of $R_l$ for the closest and the farthest user to the target BS. For the cell-interior user ($l=1$) we see that both the empirical and theoretical pdfs are rather similar; hence, we may expect that coverage results from both models are also similar (as illustrated in the previous figures). However, for cell-edge users both pdfs have different shapes. Specifically, we notice that the distances of cell-edge users in the analytical model tend to be lower than the distances in the simulation model. This explains the difference in coverage probability, specially in the absence of power control as exhibited in Fig. \ref{fig:fig_V_3}. Since the distance for the cell-edge user tends to be lower in the analytical model when power control is not used, the desired signal tends to be higher and so the coverage probability grows. This is mitigated by using power control, since this technique aims to obtain equal received power from all users independently of their positions.
\vspace{-2mm}
\begin{figure}[ht]
\centering
\includegraphics[width=2.75in]{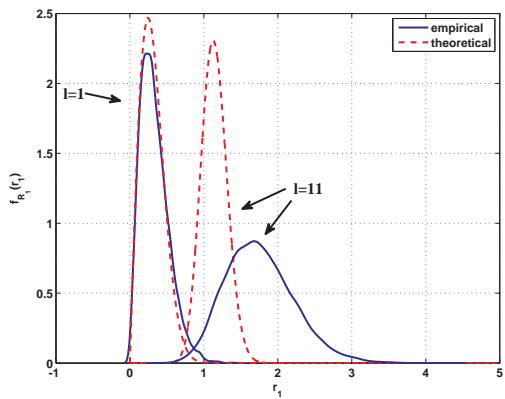}
\caption{Empirical and theoretical pdfs of $R_l$, ($l=1,k=11$) and ($l=k=11$)}
\label{fig:fig_V_5}
\end{figure}
\vspace{-2mm}
%
%\begin{figure}[!t]
%\centering
%\includegraphics[width=2.75in]{fig_V_6.eps}
%\caption{ Empirical and theoretical pdf of $R_l$ of simulation and analytical models respectively with $l=k=11$ }
%\label{fig:fig_V_6}
%\end{figure}
%
%In Fig. \ref{fig:fig_V_2} and Fig. \ref{fig:fig_V_3} the coverage probability for  $\epsilon=0.75$ and $\epsilon=0$ are plotted with $k=11$. It can be seen that the analysis matches the simulations. It can be observed that the difference between cell-interior users and cell-edge users becomes more relevant as $\epsilon$ becomes smaller. This means that the fairness between the users reduces as the power control factor is lower. Extremes cases are given by $\epsilon = 0$ and $\epsilon = 1$ where the fairness is minimal and maximal respectively.

%Finally, Fig. \ref{fig:fig_V_4} compares coverage probability for different $\epsilon$. It can be observed that as $\epsilon$ becomes lower the fairness gets reduced; however, the maximum coverage gets increased, e.g. consider user $l=1$ which has coverage probability close to 1 even for $t=10$dB. Therefore, there is a trade-off between fairness and maximal coverage in uplink cellular systems.
\vspace{-2mm}
\section{Conclusions}
\label{sec:Conclusions}
We proposed a tractable analysis model for multi-user uplink cellular networks based on conditional thinning. Assuming that there are $k$ active users scheduled on $k$ orthogonal resources, the joint distribution of the distances from the target BS to the $l^{th}$ user and to the farthest $k^{th}$ user have been obtained.  Thinning  outside the target cell with probability $1/k$ is used to obtain the actual set of interfering users. A more realistic model with BS distributed as PPP and one interfering user within its Voronoi cell has been simulated as well. Results show that fractional power control permits to increase fairness among users, at the expense of reducing the coverage probability of cell-interior users as $\epsilon$ grows. The coverage results provided by the analysis model are close to the simulation models when power control is used; the difference of behavior in the absence of power control is also discussed by studying the marginal distributions of the distances of the users to the serving BS.
%We proposed a tractable analytical model for multi-user uplink cellular networks based on conditional thinning. Assuming that there are $k$ active users scheduled in $k$ orthogonal resources, there is one out of $k$ interfering users per cell, and so the thinning probability is set to $p=1/k$.  Results show that fractional power control permits to increase fairness among users, at the expense of reducing the coverage probability of cell-interior users as fairness increases (i.e. as $\epsilon$ grows). The coverage results provided by the analytical model are close to the simulation models when power control is used; the difference of behavior in the absence of power control is also discussed by studying the marginal distributions of the distances of the users to the serving BS.
%\vspace{-2mm}
%\section*{Acknowledgment}
%
%%
%\vspace{-2mm}
\bibliographystyle{IEEEtran}
\bibliography{IEEEabrv,stochasticGeometry}

\begin{thebibliography}{10}
\providecommand{\url}[1]{#1}
\csname url@samestyle\endcsname
\providecommand{\newblock}{\relax}
\providecommand{\bibinfo}[2]{#2}
\providecommand{\BIBentrySTDinterwordspacing}{\spaceskip=0pt\relax}
\providecommand{\BIBentryALTinterwordstretchfactor}{4}
\providecommand{\BIBentryALTinterwordspacing}{\spaceskip=\fontdimen2\font plus
\BIBentryALTinterwordstretchfactor\fontdimen3\font minus
  \fontdimen4\font\relax}
\providecommand{\BIBforeignlanguage}[2]{{%
\expandafter\ifx\csname l@#1\endcsname\relax
\typeout{** WARNING: IEEEtran.bst: No hyphenation pattern has been}%
\typeout{** loaded for the language `#1'. Using the pattern for}%
\typeout{** the default language instead.}%
\else
\language=\csname l@#1\endcsname
\fi
#2}}
\providecommand{\BIBdecl}{\relax}
\BIBdecl

\bibitem{Himayat10}
N.~Himayat, S.~Talwar, A.~Rao, and R.~Soni, ``{Interference management for 4G
  cellular standards [WIMAX/LTE UPDATE]},'' \emph{IEEE Communications
  Magazine}, vol.~48, no.~8, pp. 86--92, Aug. 2010.

\bibitem{Wyner94}
A.~D. Wyner, ``{Shannon-theoretic approach to a Gaussian cellular
  multiple-access channel},'' \emph{IEEE Transactions on Information Theory},
  vol.~40, no.~6, pp. 1713--1727, 1994.

\bibitem{Tukmanov13}
A.~Tukmanov, Z.~Ding, S.~Boussakta, and A.~Jamalipour, ``On the impact of
  network geometric models on multicell cooperative communication systems,''
  \emph{IEEE Wireless Commun.}, vol.~20, no.~1, pp. 75--81, Feb. 2013.

\bibitem{Xu11}
J.~Xu, J.~Zhang, and J.~G. Andrews, ``{On the Accuracy of the Wyner Model in
  Cellular Networks},'' \emph{IEEE Transactions on Wireless Communications},
  vol.~10, no.~9, pp. 3098--3109, Sep. 2011.

\bibitem{Haenggi09}
M.~Haenggi, J.~G. Andrews, F.~Baccelli, O.~Dousse, and M.~Franceschetti,
  ``{Stochastic geometry and random graphs for the analysis and design of
  wireless networks},'' \emph{IEEE Journal on Selected Areas in
  Communications}, vol.~27, no.~7, pp. 1029--1046, Sep. 2009.

\bibitem{Andrews11}
J.~G. Andrews, F.~Baccelli, and R.~K. Ganti, ``{A Tractable Approach to
  Coverage and Rate in Cellular Networks},'' \emph{IEEE Transactions on
  Communications}, vol.~59, no.~11, pp. 3122--3134, Nov. 2011.

\bibitem{Novlan13}
T.~D. Novlan, H.~S. Dhillon, and J.~G. Andrews, ``{Analytical Modeling of
  Uplink Cellular Networks},'' \emph{IEEE Transactions on Wireless
  Communications}, vol.~12, no.~6, pp. 2669--2679, June 2013.

\bibitem{Novlan13b}
T.~D. Novlan and J.~G. Andrews, ``{Analytical Evaluation of Uplink Fractional
  Frequency Reuse},'' \emph{IEEE Transactions on Communications}, vol.~61,
  no.~5, pp. 2098--2108, May 2013.

\bibitem{Elsawy14}
H.~ElSawy and E.~Hossain, ``{On Stochastic Geometry Modeling of Cellular Uplink
  Transmission with Truncated Channel Inversion Power Control},'' \emph{arXiv
  preprint arXiv:1401.6145}, 2014.

\bibitem{Dhillon13}
H.~S. Dhillon, R.~K. Ganti, and J.~G. Andrews, ``{Modeling Non-Uniform UE
  Distributions in Downlink Cellular Networks},'' \emph{IEEE Wireless
  Communications Letters}, vol.~2, no.~3, pp. 339--342, June 2013.

\bibitem{Dhillon13b}
------, ``{Load Aware Modeling and Analysis of Heterogeneous Cellular
  Networks},'' \emph{IEEE Transactions on Wireless Communications}, vol.~12,
  no.~4, pp. 1666--1677, Apr. 2013.

\bibitem{Stoyan13}
S.~N. Chiu, D.~Stoyan, W.~S. Kendall, and J.~Mecke, \emph{{Stochastic Geometry
  and Its Applications}}.\hskip 1em plus 0.5em minus 0.4em\relax Wiley Series
  in Probability and Statistics, 2013.

\bibitem{Papoulis}
A.~Papoulis and S.~U. Pillai, \emph{Probability, Random Variables, and
  Stochastic Processes}, ser. McGraw-Hill series in electrical engineering:
  Communications and signal processing.\hskip 1em plus 0.5em minus 0.4em\relax
  Tata McGraw-Hill, 2002.

\bibitem{Haenggi05}
M.~Haenggi, ``{On distances in uniformly random networks},'' \emph{IEEE
  Transactions on Information Theory}, vol.~51, no.~10, pp. 3584--3586, Oct
  2005.

\end{thebibliography}
\end{document}